\documentclass[letterpaper, 10 pt,conference]{ieeeconf}\IEEEoverridecommandlockouts
\usepackage{graphicx}
\usepackage[T1]{fontenc}
\usepackage{siunitx}
\usepackage{lmodern}
\usepackage[english]{babel}
\usepackage{amsmath} 

\DeclareMathOperator*{\argmin}{arg\,min}
\usepackage{amsfonts}
\usepackage{amssymb}
\usepackage{units}
\usepackage{afterpage}
\usepackage{microtype}
\usepackage{gensymb}
\usepackage[hidelinks]{hyperref}
\hypersetup{
    colorlinks,
    linkcolor={black},
    citecolor={black},
    urlcolor={blue!90!black}
}
\usepackage{listings}
\usepackage{algpseudocode}
\usepackage{booktabs}\usepackage{multirow}
\usepackage{color}
\usepackage{pgfplots}
\pgfplotsset{compat=newest}
\pgfplotsset{plot coordinates/math parser=false,title style={yshift={-1.5mm}}}
\usepgfplotslibrary{groupplots}
\usetikzlibrary{plotmarks}
\usetikzlibrary{positioning}
\usetikzlibrary{shapes,arrows}
\usepgfplotslibrary{patchplots}
\definecolor{blue}{rgb}{0,0,0.55}
\definecolor{green}{rgb}{0,.9,0}
\definecolor{darkred}{rgb}{0.5,0,0}
\definecolor{yellow}{rgb}{.5,.5,0}
\definecolor{linecolor1}{rgb}{0,0,0.4}
\definecolor{linecolor2}{rgb}{1,0.55,0}
\definecolor{linecolor3}{rgb}{0.1,1,1}
\definecolor{linecolor4}{rgb}{0.6,0,0.5}

\newcommand{\recmt}[1]{}
\newcommand{\cmt}[1]{}

\newcommand{\figurecaptionreduction}{\vspace{-2.5mm}}

\newcommand{\T}{^{\hspace{-0.1mm}\scriptscriptstyle \mathsf{T}}\hspace{-0.2mm}}

\newcommand{\normt}[1]{\begin{Vmatrix}#1\end{Vmatrix}_2}
\newcommand{\norm}[1]{\begin{Vmatrix}#1\end{Vmatrix}}
\newcommand{\inspace}[1]{\in \mathbb{R}^{#1}}

\newcommand{\A}{\Phi}
\newcommand{\y}{y}
\newcommand{\PI}{\left(\A \hspace{-0.2mm}\T\hspace{-0.1mm}\A\right)^{\hspace{-0.4mm}-1} \hspace{-1mm} \A\hspace{-0.3mm}\T}

\newcommand{\w}{k}
\newcommand{\N}{\mathcal{N}}

\newcommand{\minimize}[1]{\underset{#1}{\text{minimize} }}
\newcommand{\subjto}{\text{subject to }}
\renewcommand{\vec}[1]{\operatorname{vec}{(#1)}}

\newcommand{\bmatrixx}[1]{\begin{bmatrix}#1\end{bmatrix}}

\definecolor{linecolor1}{rgb}{0,0,0.4}
\definecolor{linecolor2}{rgb}{1,0.55,0}
\definecolor{linecolor3}{rgb}{0.1,1,1}
\definecolor{linecolor4}{rgb}{0.6,0,0.5}

\title{\LARGE \bf
Identification of LTV Dynamical Models with\\ Smooth or Discontinuous Time Evolution \\by means of Convex Optimization
}

\author{
\centering Fredrik Bagge Carlson* \quad Anders Robertsson \quad Rolf Johansson
\thanks{*Open-source implementations of all presented methods and examples in this paper are made available at \href{github.com/baggepinnen/LTVModels.jl}{github.com/baggepinnen/LTVModels.jl} (Will be made available in advance of paper publication). The reported research was supported by the European Commission under the Framework Programme Horizon 2020 under grant agreement 644938 SARAFun. The authors are members of the LCCC Linnaeus Center and the eLLIIT Excellence Center at Lund University, Dept Automatic Control, Lund Sweden.\protect\\
{Fredrik.Bagge\_Carlson@control.lth.se}}
}

\usepackage[capitalise]{cleveref}
\crefname{section}{Sec.}{Sections}
\Crefname{section}{Section}{Sections}

\crefname{proposition}{Proposition}{Propositions}
\Crefname{proposition}{Proposition}{Propositions}
\crefname{corollary}{Corollary}{Corollaries}
\Crefname{corollary}{Corollary}{Corollaries}
\crefname{theorem}{Theorem}{Theorems}
\Crefname{theorem}{Theorem}{Theorems}

\begin{document}
\newtheorem{proposition}{Proposition}
\newtheorem{corollary}{Corollary}
\newtheorem{theorem}{Theorem}
\newlength\figureheight
\newlength\figurewidth
\setlength{\figurewidth}{0.4\textwidth}
\setlength{\figureheight }{4cm }

\maketitle
\thispagestyle{empty}
\pagestyle{empty}

\begin{abstract}
    We establish a connection between trend filtering and system identification which results in a family of new identification methods for linear, time-varying (LTV) dynamical models based on convex optimization. We demonstrate how the design of the cost function promotes a model with either a continuous change in dynamics over time, or causes discontinuous changes in model coefficients occurring at a finite (sparse) set of time instances. We further discuss the introduction of priors on the model parameters for situations where excitation is insufficient for identification. The identification problems are cast as convex optimization problems and are applicable to, e.g., ARX models and state-space models with time-varying parameters. We illustrate usage of the methods in simulations of jump-linear systems, a nonlinear robot arm with non-smooth friction and stiff contacts as well as in model-based, trajectory centric reinforcement learning on a smooth nonlinear system.
\end{abstract}

\section{Introduction}\label{sec:introduction}

The difficulty of the task of identifying time-varying dynamical models of systems varies greatly with the model considered and the availability of measurements of the state sequence. For smoothly changing dynamics, linear in the parameters, the recursive least-squares algorithm with exponential forgetting (RLS$\lambda$) is a common option. If a Gaussian random-walk model for the parameters is assumed, a Kalman filtering/smoothing algorithm \cite{rauch1965maximum} gives the filtering/smoothing densities of the parameters in closed form. The assumption of smoothly (Gaussian) varying dynamics is often restrictive. Discontinuous dynamics changes occur, for instance, when an external controller changes operation mode, when a sudden contact between a robot and its environment is established, an unmodeled disturbance enters the system or when a system is suddenly damaged.

Identification of systems with non-smooth dynamics evolution has been studied extensively. The book~\cite{costa2006discrete} treats the case where the dynamics are known, but the state sequence unknown, i.e., state estimation. In~\cite{nagarajaiah2004time}, the authors examine the residuals from an initial constant dynamics fit to determine regions in time where improved fit is needed by the introduction of additional constant dynamics models. Results on identifiability and observability in jump-linear systems in the non-controlled (autonomous) setting are available in~\cite{vidal2002observability}. The main result on identifiability in \cite{vidal2002observability} was a rank condition on a Hankel matrix constructed from the collected output data, similar to classical results on the least-squares identification of ARX models which appears as rank constraints on the, typically Toeplitz or block-Toeplitz, regressor matrix. Identifiability of the methods proposed in this article are discussed in~\cref{sec:identifiability}.

An LTV model can be seen as a first-order approximation of the dynamics of a nonlinear system around a trajectory. We emphasize that such an approximation will in general fail to generalize far from the this trajectory, but many methods in reinforcement learning and control make efficient use of the linearized dynamics for optimization, while ensuring validity of the approximation by constraints or penalty terms. An example provided in \cref{sec:rl} highlights such a method.

An important class of identification methods that has been popularized lately is \emph{trend filtering} methods~\cite{kim2009ell_1, tibshirani2014adaptive}. Trend filtering methods work by specifying a \emph{fitness criterion} that determines the goodness of fit, as well as a \emph{regularization} term, often chosen with sparsity promoting qualities. As a simple example, consider the reconstruction $\hat y$ of a noisy signal $y = \{y_t\inspace{}\}_{t=1}^T$ with piecewise constant segments. To this end, we may formulate and solve the convex optimization problem
\begin{equation} \label{eq:tf}
   \minimize{\hat{y}} \normt{y-\hat{y}}^2 + \lambda\sum_t |\hat{y}_{t+1} - \hat{y}_t|
\end{equation}
The first term is the fitness criterion or \emph{loss function}, whereas the second term is a sparsity-promoting regularizer which promotes small changes between consecutive samples in the reconstructed signal. The sparsity promoting effect of the 1-norm regularizer is well known, and stems from the constant length of the gradient whenever the argument is non-zero~\cite{murphy2012machine}. Compare this to the squared difference, for which the gradient rapidly vanishes as the argument approaches zero. The squared difference will thus promote \emph{small} arguments, whereas the 1-norm promotes \emph{sparse} arguments.

In this work, we will draw inspiration from the trend-filtering literature to develop new system identification methods for LTV models, with interesting properties. In trend filtering, we decompose a curve as a set of polynomial segments. In the identification methods proposed in this work, we instead decompose a multivariable state sequence as the output of a set of LTV models, where the model coefficients evolve as polynomial functions of time. We start by defining a set of optimization problems with a least-squares loss function and carefully chosen regularization terms. We further discuss how prior information can be utilized to increase the accuracy of the identification and end the article with identification of a nonlinear system with non-smooth friction and an example of model-based reinforcement learning followed by a discussion.

\section{LTI identification}\label{sec:lti}

We start by considering the case of identification of the parameters in an LTI model on the form
\begin{equation}
   x_{t+1} = A x_t + B u_t + v_t, \quad t \in [1,T]
\end{equation}
where $x\inspace{n}$, $u\inspace{m}$ are the state and input respectively. A discussion around the noise term $v_t$ is deferred until~\cref{sec:general}, where we indicate how statistical assumptions on $v_t$ influence the cost function and the properties of the estimate. If the state and input sequences are known, a plethora of methods for estimating the parameters exists. A common method for systems that are linear in the parameters is the least-squares (LS) method, which in case of Gaussian noise, $v$, coincides with the maximum likelihood (ML) estimate. To facilitate estimation using the LS method, we write the model on the form $\y = \A\w$, and arrange the data according to
\begin{align*}
   \y &=
   \begin{bmatrix}
      {x_1} \\ \vdots \\ {x_T}
   \end{bmatrix} & &\inspace{Tn} \\
   \w &= \vec{\bmatrixx{A\T & B\T}} & &\inspace{K}\\[0.2em]
   \A &=
   \begin{bmatrix}
      I_n \otimes x_0\T & I_n \otimes u_0\T \\
      \vdots & \vdots\\
      I_n \otimes x_{T-1}\T & I_n \otimes u_{T-1}\T
   \end{bmatrix}
   & &\in \mathbb{R}^{Tn\times K}
\end{align*}
where $\otimes$ denotes the Kronecker product and $K=n^2+nm$ is the number of model parameters, and solve the optimization problem \labelcref{eq:lscost} with closed-form solution \labelcref{eq:ls}.
\begin{align}
   \w^* &= \argmin_{\w} \normt{\A \w - \y}^2 \label{eq:lscost}\\
   ~ &= \PI \y \label{eq:ls}
\end{align}

\section{Time-varying dynamics}
We now move on to the contribution of this work, and extend our view to systems where the dynamics change with time. We limit the scope of this article to models on the form
\begin{equation}
\label{eq:tvk}
\begin{split}
   x_{t+1} &= A_t x_t + B_t u_t + v_t\\
   \w_t &= \vec{\bmatrixx{A_t\T & B_t\T}}
\end{split}
\end{equation}
where the parameters $\w$ are assumed to evolve according to the dynamical system
\begin{equation}
    \label{eq:dynsys}
\begin{split}
   k_{t+1} &= H_t k_t + w_t\\
   y_t &= \big(I_n \otimes \bmatrixx{x_t\T & u_t\T}\big) \w_t
\end{split}
\end{equation}
where, if no prior knowledge is available, the dynamics matrix $H_t$ can be taken as the identity matrix; $H = I$ implies that the model coefficients follow a random walk dictated by the properties of $w_t$, i.e., the state transition density function $p_w(k_{t+1}|k_t)$. The emission density function $p_v(x_{t+1} | x_t, u_t, k_t)$ is determining the drift of the state, which for the parameter estimation problem can be seen as the distribution of measurements, given the current state of the system.
We emphasize here that the state in the parameter evolution model refers to the current parameters $k_t$ and not the system state $x_t$, hence, $p_v$ is called the emission density and not the transition density. Particular choices of $p_v$ and $p_w$ emit data likelihoods concave in the parameters and hence amenable to convex optimization.

The following sections will introduce a number of optimization problems with different regularization functions, corresponding to different choices of $p_w$, and different regularization arguments, corresponding to different choices of $H$. We also discuss the quality of the identification resulting from the different modeling choices.

\subsection{Low frequency time evolution}
A slowly varying signal is characterized by \emph{small first-order time differences}.
To identify slowly varying dynamics parameters, we thus penalize the squared 2-norm of the first-order time difference of the model parameters, and solve the optimization problem
\begin{equation} \label{eq:slow}
   \minimize{\w} \normt{\y-\hat{\y}}^2 + \lambda^2\sum_t \normt{\w_{t+1} - \w_{t}}^2
\end{equation}
where $\sum_t$ denotes the sum over relevant indices $t$, in this case $t\in [1,T-1]$.
This optimization problem has a closed form solution given by
\begin{align}\label{eq:closedform}
    \tilde{\w}^* &= (\tilde{\A}\T\tilde{\A} + \lambda^2 D_1\T D_1)^{-1}\tilde{\A}\T \tilde{Y}\\
    \tilde{\w} &= \operatorname{vec}(\w_1, ...\,, \w_T)\nonumber
\end{align}
where $\tilde{\A}$ and $\tilde{Y}$ are appropriately constructed matrices and the first-order differentiation operator matrix $D_1$ is constructed such that $\lambda^2\normt{D_1 \tilde{\w}}^2$ equals the second term in~\labelcref{eq:slow}.
The computational complexity $\mathcal{O}\big((TK)^3\big)$ of computing $k^*$ using the closed-form solution \labelcref{eq:closedform} becomes prohibitive for all but toy problems. We note that the cost function in \labelcref{eq:slow} is the negative data log-likelihood of a Brownian random-walk parameter model with $H=I$, which motivates us to develop a dynamic programming algorithm based on a Kalman smoother, detailed in~\cref{sec:kalmanmodel}.

\subsection{Smooth time evolution}
A smoothly varying signal is characterized by \emph{small second-order time differences}.
To identify smoothly time-varying dynamics parameters, we thus penalize the squared 2-norm of the second-order time difference of the model parameters, and solve the optimization problem
\begin{equation} \label{eq:smooth}
   \minimize{\w} \normt{\y-\hat{\y}}^2 + \lambda^2\sum_t \normt{\w_{t+2} -2 \w_{t+1} + \w_t}^2
\end{equation}
Also this optimization problem has a closed form solution on the form~\labelcref{eq:closedform} with the corresponding second-order differentiation operator $D_2$. \Cref{eq:smooth} is the negative data log-likelihood of a Brownian random-walk parameter model with added momentum and $H$ derived in~\cref{sec:smoothkalman}, where a Kalman smoother with augmented state is developed to find the optimal solution. We also extend problem \labelcref{eq:smooth} to more general regularization terms in \cref{sec:kalmanmodel}.

\subsection{Piecewise constant time evolution}\label{sec:pwconstant}
In the presence of discontinuous or abrupt changes in the dynamics, estimation method~\labelcref{eq:smooth} might perform poorly. A signal which is mostly flat, with a small number of distinct level changes, is characterized by a \emph{sparse first-order time difference}. To detect sudden changes in dynamics, we thus formulate and solve the problem
\begin{equation} \label{eq:pwconstant}
   \minimize{\w} \normt{\y-\hat{\y}}^2 + \lambda\sum_t \normt{ \w_{t+1} - \w_t}
\end{equation}
We can give \labelcref{eq:pwconstant} an interpretation as a \emph{grouped-lasso} cost function, where instead of groups being formed out of variables, our groups are defined by differences between variables. We thus have a penalty on the 1-norm on the \emph{length} of the difference vectors $\w_{t+1} - \w_t$ since $\norm{\normt{\cdot}}_1 = \normt{\cdot}$.
The 1-norm is a \emph{sparsity-promoting} penalty, hence a solution in which only a small number of non-zero first-order time differences in the model parameters is favored, i.e., a piecewise constant dynamics evolution.
At a first glance, one might consider the formulation
\begin{equation} \label{eq:pwconstant_naive}
   \minimize{\w} \normt{\y-\hat{\y}}^2 + \lambda\sum_t \norm{\w_{t+1} - \w_t}_1
\end{equation}
which results in a dynamics evolution with sparse changes in the coefficients, but changes to different entries of $\w_t$ are not necessarily occurring at the same time instants. The formulation~\labelcref{eq:pwconstant}, however, promotes a solution in which the change occurs at the same time instants for all coefficients in $A$ and $B$, i.e., $\w_{t+1} = \w_{t}$ for most $t$.

\subsubsection{Implementation}
Due to the non-squared norm penalty $\sum_t \normt{ \w_{t+1} - \w_t}$, problem \labelcref{eq:pwconstant} is significantly harder to solve than \labelcref{eq:smooth}. An efficient implementation using the linearized ADMM algorithm \cite{parikh2014proximal} is made available in the accompanying repository.

\subsection{Piecewise constant time evolution with known number of steps}

If the number of switches in dynamics parameters, $M$, is known in advance, the optimal problem to solve is
\begin{align} \label{eq:l0}
   & \minimize{\w} &  & \normt{\y-\hat{\y}}^2\\
   & \subjto &        & \sum_t \textbf{1}\{ \w_{t+1} \neq \w_t\} \leq M
\end{align}
where $\textbf{1}\{\cdot\}$ is the indicator function.
This problem is non-convex and we propose solving it using dynamic programming (DP). For this purpose we modify the algorithm developed in \cite{bellman1961approximation}, an algorithm frequently referred to as segmented least-squares~\cite{bellman1969curve}, where a curve is approximated by piecewise linear segments. The modification lies in the association of each segment (set of consecutive time indices during which the parameters are constant) with a dynamics model, as opposed to a simple straight line.\footnote{Indeed, if a simple integrator is chosen as dynamics model and a constant input is assumed, the result of our extended algorithm reduces to the segmented least-squares solution.} Unfortunately, the computational complexity of the dynamic programming solution, $\mathcal{O}(T^2K^3)$, becomes prohibitive for large $T$.\footnote{For details regarding the DP algorithm and implementation, the reader is referred to the source-code repository accompanying this article.}

\subsection{Piecewise linear time evolution}\label{sec:pwlinear}
A piecewise linear signal is characterized by a \emph{sparse second-order time difference}, i.e., it has a small number of changes in the slope. A piecewise linear time-evolution of the dynamics parameters is hence obtained if we solve the optimization problem.
\begin{equation} \label{eq:pwlinear}
   \minimize{\w} \normt{\y-\hat{\y}}^2 + \lambda\sum_t \normt{\w_{t+2} -2 \w_{t+1} + \w_t}
\end{equation}

\subsection{Summary}
The proposed optimization problems are summarized in~\cref{tab:opts}. The table illustrates how the choice of regularizer and order of time-differentiation of the parameter vector affects the quality of the resulting solution.

\begin{table}[]
\centering
\caption{Summary of optimization problem formulations. $D_n$ refers to parameter vector time-differentiation of order $n$.}
\label{tab:opts}
\begin{tabular}{@{}lll@{}}
\toprule
Norm & $D_n$ & Result    \\ \midrule
1    & 1         & Small number of steps (piecewise constant)     \\
1    & 2         & Small number of bends (piecewise affine)    \\
2    & 1         & Small steps (slowly varying) \\
2    & 2         & Small bends (smooth)  \\ \bottomrule
\end{tabular}
\end{table}

\subsection{Two-step refinement}

Since many of the proposed formulations of the optimization problem penalize the size of the changes to the parameters, solutions in which the changes are slightly underestimated are favored. To mitigate this issue, a two-step procedure can be implemented where in the first step, change points (knots) are identified. In the second step, the penalty on the one-norm is removed and equality constraints are introduced between consecutive time-indices for which no change in dynamics was indicated by the first step.

The second step can be computed very efficiently by noticing that the problem can be split into several identical sub-problems at the knots identified in the first step. The sub-problems have closed-form solutions if the problem in~\cref{sec:pwconstant} is considered.

To identify the points at which the dynamics change, we observe the argument inside the sum of the regularization term, i.e., $a_{t1} = \normt{ \w_{t+1} - \w_t}$ or $a_{t2} = \normt{\w_{t+2} -2 \w_{t+1} + \w_t}$. Time instances where $a_t$ is taking non-zero values indicate change points.

\section{Dynamics prior and Kalman filtering}
\recmt{Reviewer found this section rushed}
The identifiability of the parameters in a dynamical model hinges on the observability of the dynamics system~\labelcref{eq:dynsys}, or more explicitly, only modes excited by the input $u$ will be satisfactorily identified. If the identification is part of an iterative learning and control scheme, e.g., ILC or reinforcement learning, it might be undesirable to introduce additional noise in the input to improve excitation for identification. This section will introduce prior information about the dynamics which mitigates the issue of poor excitation of the system modes. The prior information might come from, e.g., a nominal model known to be inaccurate, or an estimated global model such as a Gaussian mixture model (GMM). A statistical model of the joint density $p(x_{t+1},x_t,u_t)$ constructed from previously collected tuples $(x_{t+1},x_t,u_t)$ provides a dynamical model of the system through the conditional pdf $p(x_{t+1}|x_t,u_t)$.

We will see that for priors from certain families, the resulting optimization problem remains convex. For the special case of a Gaussian prior over the dynamics parameters or the output, the posterior mean of the parameter vector is conveniently obtained from a Kalman-smoothing algorithm, modified to include the prior.

\subsection{General case}\label{sec:general}

If we introduce a parameter state $\w$ (c.f., \labelcref{eq:tvk}) and a prior over all parameter-state variables $p(\w_{t}|z_t)$, where the variable $z_t$ might be, for instance, the time index $t$ or state $x_t$, we have the data log-likelihood
\begin{equation}\label{eq:lcdll}
    \begin{split}
    \log p(\w,y|x,z)_{1:T} &= \sum_{t=1}^T \log p(y_t|\w_t,x_t) \\
    + \sum_{t=1}^{T-1} \log p(\w_{t+1}|\w_t) &+ \sum_{t=1}^{T} \log p(\w_{t}|z_t)
\end{split}
\end{equation}
which factors conveniently due to the Markov property of a state-space model.
For particular choices of density functions in~\labelcref{eq:lcdll}, notably Gaussian and Laplacian, the negative likelihood function becomes convex. The next section will elaborate on the Gaussian case and introduce a recursive algorithm that solves for the full posterior efficiently. The Laplacian case, while convex, does not admit an equally efficient algorithm, but is more robust to outliers in the data.

\subsection{Gaussian case} \label{sec:kalmanmodel}
If all densities in~\labelcref{eq:lcdll} are Gaussian and $\w$ is modeled with the Brownian random walk model \labelcref{eq:dynsys} (Gaussian $v_t$), \labelcref{eq:lcdll} can be written on the form (scaling constants omitted)
\begin{equation}
    \begin{split}\label{eq:prioropt}
    -\log p(\w,y|x,z)_{1:T} &=
    \sum_{t=1}^T \norm{y_t - \hat{y}(\w_t,x_t)}^2_{\Sigma^{-1}_y} \\
    &+ \sum_{t=1}^{T-1} \norm{\w_{t+1} - \w_{t}}^2_{\Sigma^{-1}_\w} \\
    &+ \sum_{t=1}^{T} \norm{\mu_0(z_t) - \w_{t}}^2_{\Sigma^{-1}_0(z_t)}
    \end{split}
\end{equation}
for some function $\mu_0(z_t)$ which produces the prior mean of $\w$ given $z_t$. $\Sigma_y, \Sigma_\w, \Sigma_0(z_t)$ are the covariance matrices of the state-drift, parameter drift and prior respectively and $\norm{x}_{\Sigma^{-1}}^2 = x\T\Sigma^{-1}x$.

In this special case, we introduce a recursive solution given by a modified Kalman smoothing algorithm, where the conditional mean of the state is updated with the prior. Consider the standard Kalman filtering equations, reproduced here to establish the notation
    \begin{align}
        \hat{x}_{t|t-1} &= A \hat{x}_{t-1|t-1} + B u_{t-1} \\
        P_{t|t-1} &= A P_{t-1|t-1}A\T + R_1\\
        K_t &= P_{t|t-1}C\T\big(CP_{t|t-1}C\T+R_2\big)^{-1}\\
        \hat{x}_{t|t} &= \hat{x}_{t|t-1} + K_t\big(y_t-C\hat{x}_{t|t-1}\big)\\
        P_{t|t} &= P_{t|t-1} - K_t C P_{t|t-1}
    \end{align}
    where $x$ is the state vector, with state-drift covariance $R_1$ and $C$ is a matrix that relates $x$ to a measurement $y=Cx$ with covariance $R_2$. The first two equations constitute the \emph{prediction} step, and the last two equations incorporate the measurement $y_t$ in the \emph{correction} step.
    The modification required to incorporate a Gaussian prior on the state variable $p(x_t|v_t) = \N(\mu_0(v_t), \Sigma_0(v_t))$ involves a repeated correction step and takes the form
    \begin{align}
        \bar{K}_t &= P_{t|t}\big(P_{t|t}+\Sigma_0(v_t)\big)^{-1}\label{eq:postk}\\
        \bar{x}_{t|t} &= \hat{x}_{t|t} + \bar{K}_t\big(\mu_0(v_t)-\hat{x}_{t|t}\big)\label{eq:postx}\\
        \bar{P}_{t|t} &= P_{t|t} - \bar{K}_t P_{t|t}\label{eq:postcov}
    \end{align}
    where $\bar{\cdot}$ denotes the posterior value. This additional correction can be interpreted as receiving a second measurement $\mu_0(v_t)$ with covariance $\Sigma_0(v_t)$.
    For the Kalman-smoothing algorithm, $\hat{x}_{t|t}$ and $P_{t|t}$ in \labelcref{eq:postx,eq:postcov} are replaced with $\hat{x}_{t|T}$ and $P_{t|T}$.

A prior over the output of the system, or a subset thereof, is straight forward to include in the estimation by means of an extra update step, with $C,R_2$ and $y$ being replaced with their corresponding values according to the prior.

\subsection{Kalman filter for identification}
We can employ the Kalman-based algorithm to solve two of the proposed optimization problems:
\subsubsection{Low frequency}
The Kalman smoother can be used for solving identification problems like~\labelcref{eq:slow} by noting that~\labelcref{eq:slow} is the negative log-likelihood of the dynamics model~\labelcref{eq:dynsys}. The identification problem is thus reduced to a standard state-estimation problem.
\subsubsection{Smooth}\label{sec:smoothkalman}

To develop a Kalman-filter based algorithm for solving~\labelcref{eq:smooth}, we augment the model \labelcref{eq:dynsys} with the state variable $k^\prime_t = k_{t} - k_{t-1}$ and note that $k^\prime_{t+1} - k^\prime_t = k_{t+1} - 2k_t + k_{t-1}$. We thus introduce the augmented-state model
\begin{align}\label{eq:dynsys2}
   \begin{bmatrix}k_{t+1} \\ k^\prime_{t+1}\end{bmatrix} &= \begin{bmatrix}I_K & I_K \\ 0_K & I_K\end{bmatrix} \begin{bmatrix}k_{t} \\ k^\prime_{t}\end{bmatrix} + \begin{bmatrix}0 \\ w_{t}\end{bmatrix}\\
   y_t &=  \begin{bmatrix}\big(I_n \otimes \bmatrixx{x_t\T & u_t\T}\big) & 0\end{bmatrix} \begin{bmatrix}k_{t} \\ k^\prime_{t}\end{bmatrix}
\end{align}
which is on a form suitable for filtering/smoothing with the machinery developed above.

\subsubsection{General case}
The Kalman-filter based identification method can be generalized to solving optimization problems where the argument in the regularizer appearing in \labelcref{eq:smooth} is replaced by a general linear operation on the parameter vector, $P(z)k$, and we have the following proposition
\begin{proposition}
    Any optimization problem on the form
    \begin{equation} \label{eq:generalkalman}
       \minimize{\w} \normt{\y-\hat{\y}}^2 + \lambda^2\sum_t \normt{P(z)\w_t}^2
    \end{equation}
    where $P(z)$ is a polynomial of degree $n>0$ in the time difference operator $z$ with $z^{-n} P(1) = 0$, can be solved with a Kalman smoother employed to an autonomous state-space system.
\end{proposition}
\begin{proof}
    Let $P^*(z^{-1}) = z^{-n} P(z)$. We assume without loss of generality that $P^*(0) = 1$ since any constant $P^*(0)$ can be factored out of the polynomial. $Q(z^{-1}) = P^*(z^{-1})^{-1}$ is a strictly proper transfer function and has a realization as a linear, Gaussian state-space system of degree $n$. Since $Q(z^{-1})$ is strictly proper, the realization has no direct term.
    The negative data log-likelihood of $Q(z^{-1})$ is equal, up to constants idenpendent of $\w$, to the cost function in~\cref{eq:generalkalman}, hence the Kalman smoother applied to $Q$ optimizes \cref{eq:generalkalman}.
\end{proof}
For~\labelcref{eq:smooth} $P(z)$ equals $z^2 - 2z + 1$ and $Q(z^{-1})$ has a realization on the form~\labelcref{eq:dynsys2}.

\section{Well-posedness and identifiability}\label{sec:identifiability}
To assess the well-posedness of the proposed identification methods, we start
by noting that the problem of finding $A$ in $x_{t+1} = Ax_t$ given a pair
$(x_{t+1},x_t)$ is an ill-posed problem in the sense that the solution is non
unique. If we are given several pairs $(x_{t+1},x_t)$, for different $t$, while
$A$ remains constant, the problem becomes over-determined and well-posed in the
least-squares sense, provided that the vectors $\{x_t^{(i)}\}_{t=1}^T$ span $\mathbb{R}^n$. The LTI-case in~\cref{sec:lti} is well posed according to
classical results, when $\Phi$ has full column rank.

When we extend our view to LTV models, the number of free
parameters is increased significantly, and the corresponding
regressor matrix $\tilde{\A}$ will never have full column
rank and the introduction of a regularization term is necessary.
Informally, for every $n$ measurements, we have $K=n^2+nm$
free parameters. If we consider the identification problem
of~\cref{eq:pwconstant} and let
$\lambda \rightarrow \infty$, the regularizer terms essentially becomes equality constraints.
This will enforce a solution in which all parameters
in $k$ are constant over time, and the problem reduces
to the LTI-problem. As $\lambda$ decreases,
the effective number of free parameters increases until the
problem gets ill-posed for $\lambda = 0$.
We formalize the above arguments as
\begin{proposition}
    Optimization problems \labelcref{eq:slow,eq:pwconstant} have unique global minima for $\lambda > 0$ if and only if the corresponding LTI optimization problem has a unique solution.
\end{proposition}
\begin{proof}
    The cost function is a sum of two convex terms. For a global minimum to be non-unique, the Hessians of the two terms must have intersecting nullspaces.
    In the limit $\lambda \rightarrow \infty$ the problem reduces to the LTI problem. The nullspace of the regularization Hessian, which is invariant to $\lambda$, does thus not share any directions with the nullspace of $\tilde{\A}\T \tilde{\A}$ which establishes the equivalence of identifiability between the LTI problem and the LTV problems.
\end{proof}
\begin{proposition}
    Optimization problems \labelcref{eq:smooth,eq:pwlinear} with higher order differentiation in the regularization term have unique global minima for $\lambda > 0$ if and only if there exists no vector $v \inspace{n+m}$ such that
    \begin{equation}
        C^{xu}_t v = \bmatrixx{x_tx_t\T & x_tu_t\T \\ u_tx_t\T & u_tu_t\T}v = 0 \;\forall t
    \end{equation}
\end{proposition}
\begin{proof}
    Again, the cost function is a sum of two convex terms and for a global minimum to be non-unique, the Hessians of the two terms must have intersecting nullspaces.
    In the limit $\lambda \rightarrow \infty$ the regularization term reduces to a linear constraint set, allowing only parameter vectors that lie along a line through time. Let $\tilde{v} \neq 0$ be such a vector, parametrized by $t$ as $\tilde{v} = \bmatrixx{\bar{v}\T & 2\bar{v}\T & \cdots & T\bar{v}\T}\T \inspace{TK}$ where $\bar{v} = \operatorname{vec}(\left\{v\right\}_1^N) \inspace{K}$ and $v$ is an arbitrary vector $\inspace{n+m}$. $\tilde{v} \in \operatorname{null}{(\tilde{\A}\T \tilde{\A})}$ implies that the loss is invariant to the pertubation $\alpha \tilde{v}$ to $\tilde{k}$ for an arbitrary $\alpha \inspace{}$.  $(\tilde{\A}\T \tilde{\A})$ is given by $\operatorname{blkdiag}(\left\{I_n \otimes C^{xu}_t \right\}_1^T)$ which means that $\tilde{v} \in \operatorname{null}{(\tilde{\A}\T \tilde{\A})} \Longleftrightarrow  \alpha t (I_n \otimes C^{xu}_t) \bar{v} = 0 \; \forall (\alpha,t) \Longleftrightarrow \bar{v} \in \operatorname{null}{(I_n \otimes C^{xu}_t)} \;\forall t$, which implies $v \in \operatorname{null}{C^{xu}_t}$ due to the block-diagonal nature of $I_n \otimes C^{xu}_t$

\end{proof}

For the LTI problem to be well-posed, the system must be identifiable and the input $u$ must be persistently exciting of sufficient order~\cite{johansson1993system}.


\section{Example -- Jump-linear system}
\cmt{Add two link example where a stiff contact is established and there is a sign change in control signal and hence friction. Apply pwconstant method and maybe scrap linear system example. Fig 3 can be reproduced with analytical jacobian as comparison. Mention in abstract that we then evaluate on one non-linear but smooth system (pendcart) and one non-smooth system (twolink).}
\recmt{Only linear systems considered (pend-cart actually not linear.} \cmt{Highligt that pendcart is nonlinear or introduce simple robotic example. Mention early that an LTV system along a trajectory does not generalize far for a nonlinear system, hence the use of KL constraint.}

We now consider a simulated example. We generate a state sequence from the following LTV system, where the change in dynamics, from
$$A_t = \left[
\begin{array}{cc}
0.95 & 0.1 \\
0.0 & 0.95 \\
\end{array}
\right], \quad B_t = \left[
\begin{array}{c}
0.2 \\
1.0 \\
\end{array}
\right]
$$
to $$A_t = \left[
\begin{array}{cc}
0.5 & 0.05 \\
0.0 & 0.5 \\
\end{array}
\right], \quad B_t = \left[
\begin{array}{c}
0.2 \\
1.0 \\
\end{array}
\right]
$$
occurred at $t=200$.
The input was Gaussian noise of zero mean and unit variance, state transition noise and measuremet noise of zero mean and $\sigma = 0.2$ were added.
\Cref{fig:ss} depicts the estimated coefficients in the dynamics matrices for a value of $\lambda$ chosen using the L-curve method~\cite{hansen1994regularization}.
\begin{figure}
    \centering
    \setlength{\figurewidth}{0.99\linewidth}
    \setlength{\figureheight }{5.5cm}
    \input{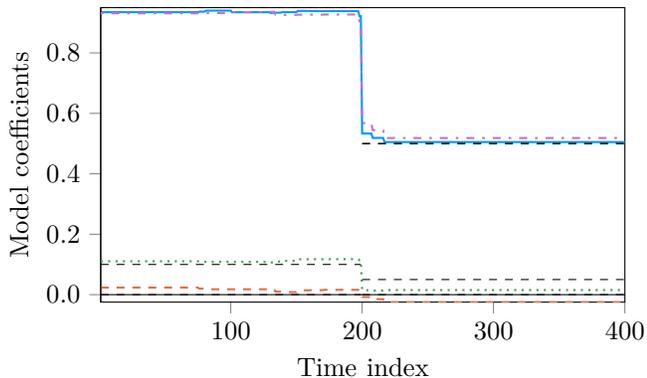}
    \figurecaptionreduction
    \caption{Piecewise constant state-space dynamics. True values are shown with dashed, black lines. Gaussian state-transition and measurement noise with $\sigma = 0.2$ were added. \recmt{Hard to read, maybe showing error is better.}}
    \label{fig:ss}
\end{figure}

\section{Example -- Non-smooth robot arm with stiff contact}
To illustrate the ability of the proposed models to represent the non-smooth dynamics along a trajectory of a robot arm, we simulate a two-link robot with discontinuous Coulomb friction. We also let the robot establish a stiff contact with the environment to illustrate both strengths and weaknesses of the modeling approach.

The state of the robot arm consists of two joint coordinates, $q$, and their time derivatives, $\dot q$. \Cref{fig:robot_train} illustrates the state trajectories, control torques and simulations of a model estimated by solving~\labelcref{eq:pwconstant}. The figure clearly illustrates that the model is able to capture the dynamics both during the non-smooth sign change of the velocity, but also during establishment of the stiff contact. The learned dynamics of the contact is however time-dependent, which is illustrated in \Cref{fig:robot_val}, where the model is used on a validation trajectory where a different noise sequence was added to the control torque. Due to the novel input signal, the contact is established at a different time-instant and as a consequence, there is an error transient in the simulated data.
\begin{figure*}[htp]
    \centering
    \setlength{\figurewidth}{0.495\linewidth}
    \setlength{\figureheight }{4cm}
    \input{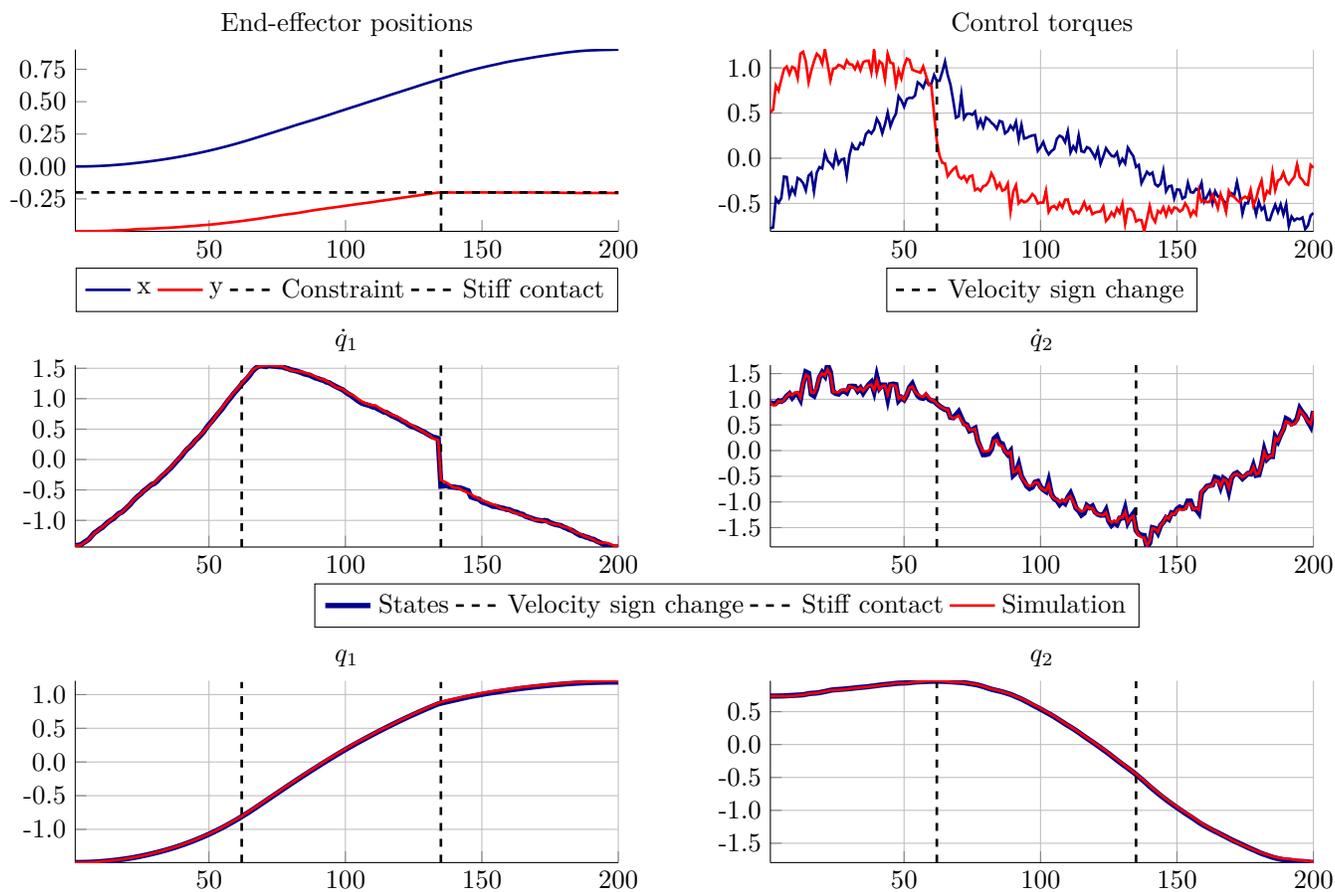}
    \caption{Simulation of non-smooth robot dynamics with stiff contact -- training data vs. sample time index. The sign change in velocity, and hence a discontinuous change in friction torque, occurs in the time interval 50-100 and the contact is established in the time interval 100-150. For numerical stability, all time-series are normalized to zero mean and unit variance, hence, the original velocity zero crossing is explicitly marked with a dashed line. The control signal plot clearly indicates the discontinuity in torque around the unnormalized zero crossing of $\dot{q}_2$.}
    \label{fig:robot_train}
\end{figure*}
\begin{figure*}[htp]
    \centering
    \setlength{\figurewidth}{0.495\linewidth}
    \setlength{\figureheight }{4cm}
    \input{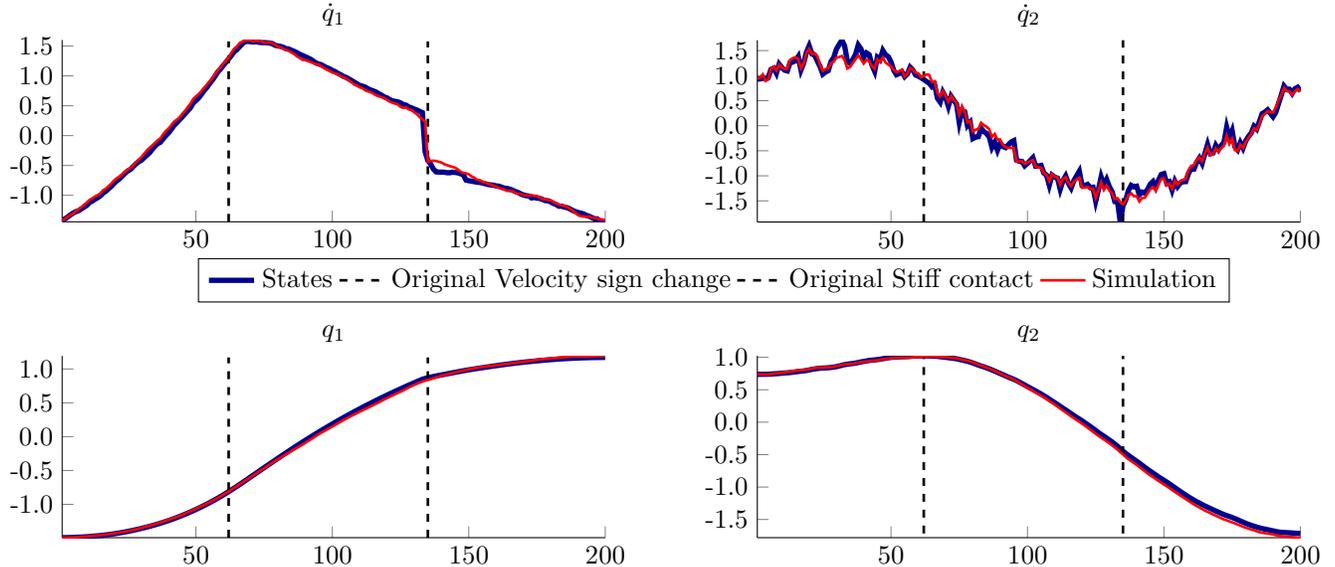}
    \caption{Simulation of non-smooth robot dynamics with stiff contact -- validation data vs. sample time index. The dashed lines indicate the event times for the training data, highlighting that the model is able to deal effortless with the non-smooth friction, but inaccurately predicts the time evolution around the contact event which now occurs at a slightly different time instance.}
    \label{fig:robot_val}
\end{figure*}

\section{Example -- Reinforcement learning} \label{sec:rl}
\recmt{More details needed}
In this example, we use the proposed methods to identify LTV dynamics models for
reinforcement learning. The goal of the task is to dampen oscillations of a pendulum attached to a moving
cart by means of moving the cart, with bounds on the control signal and a quadratic
cost on states and control. Due to the nonlinear nature of the pendulum dynamics,
linear expansions of the dynamics in the upward (initial) position and downward (final) position have poles
on opposite sides of the imaginary axis. To this end, we employ a reinforcement-learning framework inspired by~\cite{levine2013guided}, where we perform a series of rollouts whereafter each we
1) fit a dynamics model along the last obtained trajectory,
2) optimize the cost function under
the model using iterative LQG (differential dynamic programming),\footnote{Implementation made available at
\href{github.com/baggepinnen/DifferentialDynamicProgramming.jl}{github.com/baggepinnen/DifferentialDynamicProgramming.jl}} an algorithm that calculates the value function exactly under the LTV dynamics and a quadratic expansion of the cost function.
In order to stay close to the validity region of the linear model, we put bounds on the deviation between each new trajectory and the last trajectory.
We compare three different models;
the ground truth system model, an LTV model (obtained by solving \labelcref{eq:smooth}) and an LTI model.
The total cost over $T=500$ time steps is shown as a function of learning iteration
in~\cref{fig:ilc}. The figure illustrates how the learning procedure reaches the
optimal cost of the ground truth model when an LTV model is used, whereas when
using an LTI model, the learning diverges. The figure further illustrates that if
the LTV model is fit using a prior (\cref{sec:kalmanmodel}), the learning speed is increased. The prior in this case was constructed from the true system model, linearized around the last trajectory. This strategy is unavailable in a real application, but the experiment serves as an indication of the effectiveness of inclusion of a prior in this example. Future work is targeting the incremental estimation of these priors.

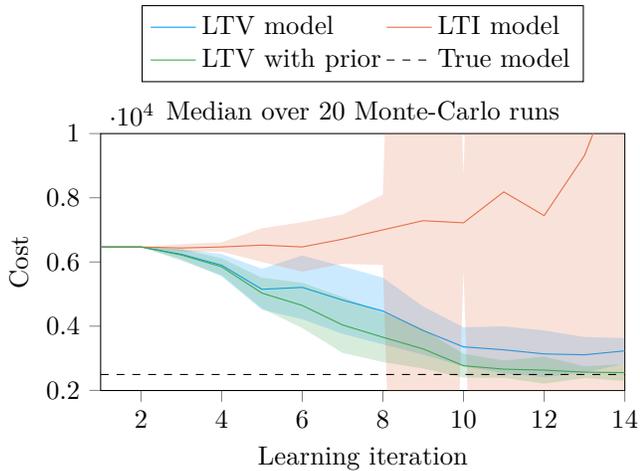
\begin{figure}[htp]
    \centering
    \setlength{\figurewidth}{0.99\linewidth}
    \setlength{\figureheight }{5cm}
\begin{tikzpicture}

\definecolor{color1}{rgb}{0.88887350027252,0.43564919034819,0.278122936141944}
\definecolor{color0}{rgb}{0,0.605603161175225,0.978680117569607}
\definecolor{color2}{rgb}{0.242224297852199,0.64327509315763,0.304448651534115}

\begin{axis}[
title={Median over 20 Monte-Carlo runs},
xlabel={Learning iteration},
ylabel={Cost},
xmin=1, xmax=14,
ymin=1995.4298956103, ymax=10000,
width=\figurewidth,
height=\figureheight,
tick align=outside,
tick pos=left,
legend style={at={(0.5,1.2)},anchor=south,legend columns=2},
legend cell align={left},
legend entries={{LTV model},{LTI model},{LTV with prior},{True model}}
]
\path [fill=color0, fill opacity=0.2] (axis cs:1,6467.89336667288)
--(axis cs:1,6467.74785554)
--(axis cs:2,6467.71983149019)
--(axis cs:3,6083.79304674676)
--(axis cs:4,5552.32630040257)
--(axis cs:5,4508.85503844202)
--(axis cs:6,4216.56451313534)
--(axis cs:7,3760.41973254615)
--(axis cs:8,3433.07343245832)
--(axis cs:9,3114.34845851781)
--(axis cs:10,2746.73053116261)
--(axis cs:11,2532.72928606864)
--(axis cs:12,2406.54154967724)
--(axis cs:13,2554.63848084459)
--(axis cs:14,2830.53984180791)
--(axis cs:14,3631.13890488011)
--(axis cs:14,3631.13890488011)
--(axis cs:13,3660.91540210268)
--(axis cs:12,3865.78714095652)
--(axis cs:11,3995.05240422497)
--(axis cs:10,3964.06066529291)
--(axis cs:9,4614.54536670643)
--(axis cs:8,5511.37483937252)
--(axis cs:7,5865.43395337106)
--(axis cs:6,6200.5529950399)
--(axis cs:5,5789.25942173269)
--(axis cs:4,6247.14060160901)
--(axis cs:3,6398.2096054346)
--(axis cs:2,6467.88160955626)
--(axis cs:1,6467.89336667288)
--cycle;

\addplot [color0]
table {%
1 6467.82061110644
2 6467.80072052323
3 6241.00132609068
4 5899.73345100579
5 5149.05723008736
6 5208.55875408762
7 4812.9268429586
8 4472.22413591542
9 3864.44691261212
10 3355.39559822776
11 3263.89084514681
12 3136.16434531688
13 3107.77694147363
14 3230.83937334401
};
\path [fill=color1, fill opacity=0.2] (axis cs:1,6467.85463420357)
--(axis cs:1,6467.74850354936)
--(axis cs:2,6467.30991947214)
--(axis cs:3,6315.12085985469)
--(axis cs:4,6322.32194235272)
--(axis cs:5,6000.52285532367)
--(axis cs:6,5697.52471480955)
--(axis cs:7,5943.29712499184)
--(axis cs:8,5906.22230384635)
--(axis cs:9,-32873.1990217889)
--(axis cs:10,5736.12520079924)
--(axis cs:11,-32392.8300679889)
--(axis cs:12,-46484.2643515633)
--(axis cs:13,-49117.1432332289)
--(axis cs:14,-48157.2936049971)
--(axis cs:14,73205.8855164305)
--(axis cs:14,73205.8855164305)
--(axis cs:13,67748.8108090554)
--(axis cs:12,61370.2531043991)
--(axis cs:11,48760.1554319038)
--(axis cs:10,8703.10980038211)
--(axis cs:9,47445.9095685239)
--(axis cs:8,8094.20785346361)
--(axis cs:7,7478.17017047317)
--(axis cs:6,7239.24244180404)
--(axis cs:5,7049.75353858485)
--(axis cs:4,6613.39519683911)
--(axis cs:3,6556.60190224514)
--(axis cs:2,6467.76558359311)
--(axis cs:1,6467.85463420357)
--cycle;

\addplot [color1]
table {%
1 6467.80156887647
2 6467.53775153262
3 6435.86138104992
4 6467.85856959592
5 6525.13819695426
6 6468.38357830679
7 6710.7336477325
8 7000.21507865498
9 7286.35527336752
10 7219.61750059068
11 8183.66268195745
12 7442.99437641792
13 9315.83378791325
14 12524.2959557167
};

\path [fill=color2, fill opacity=0.2] (axis cs:1,6467.8802891823)
--(axis cs:1,6467.68949345681)
--(axis cs:2,6467.59183527688)
--(axis cs:3,6039.12746559762)
--(axis cs:4,5587.67812664723)
--(axis cs:5,4554.0255889356)
--(axis cs:6,3939.04480375798)
--(axis cs:7,3168.2224702434)
--(axis cs:8,2882.6438663892)
--(axis cs:9,2680.56742721808)
--(axis cs:10,2395.22970651177)
--(axis cs:11,2392.21668718494)
--(axis cs:12,2208.19522012647)
--(axis cs:13,2384.14089194626)
--(axis cs:14,2293.78277584915)
--(axis cs:14,2812.73849565972)
--(axis cs:14,2812.73849565972)
--(axis cs:13,2750.26609376802)
--(axis cs:12,3053.5654671019)
--(axis cs:11,2928.96228230061)
--(axis cs:10,3140.37684237313)
--(axis cs:9,3901.14868702784)
--(axis cs:8,4429.36359425513)
--(axis cs:7,4908.14932963442)
--(axis cs:6,5351.65114546336)
--(axis cs:5,5506.64126115527)
--(axis cs:4,6114.66758111555)
--(axis cs:3,6396.98999231075)
--(axis cs:2,6467.8787317325)
--(axis cs:1,6467.8802891823)
--cycle;

\addplot [color2]
table {%
1 6467.78489131955
2 6467.73528350469
3 6218.05872895418
4 5851.17285388139
5 5030.33342504543
6 4645.34797461067
7 4038.18589993891
8 3656.00373032216
9 3290.85805712296
10 2767.80327444245
11 2660.58948474277
12 2630.88034361418
13 2567.20349285714
14 2553.26063575443
};
\addplot [black, dashed]
table {%
1 2494.28736951288
14 2494.28736951288
};
\end{axis}

\end{tikzpicture}
    \caption{Reinforcement learning example. Three different model types are used to iteratively optimize the trajectory of a pendulum on a cart. Due to the nonlinear nature of the pendulum dynamics, linear expansions of the dynamics in the upward and downward positions have poles on opposite sides of the imaginary axis, why the algorithm fails with an LTI model.}
    \label{fig:ilc}
\end{figure}

\section{Discussion}

This article presents methods for estimation of linear, time-varying models. The methods presented extend directly to nonlinear models that remain \emph{linear in the parameters}. When estimating an LTV model from a trajectory obtained from a nonlinear system, one is effectively estimating the linearization of the system around that trajectory. A first-order approximation to a nonlinear system is not guaranteed to generalize well as deviations from the trajectory become large. Many non-linear systems are, however, approximately \emph{locally} linear, such that they are well described by a linear model in a small neighborhood around the linearization/operating point. For certain methods, such as iterative learning control and trajectory centric reinforcement learning, a first-order approximation to the dynamics is used for efficient optimization, while the validity of the approximation is ensured by incorporating penalties or constraints between two consecutive trajectories.

The methods presented allow very efficient learning of this first-order approximation due to the prior belief over the nature of the change in dynamics parameters, encoded by the regularization terms. By postulating a prior belief that the dynamics parameters change in a certain way, less demand is put on the data required for identification. The identification process will thus not interfere with normal operation in the same way as if excessive noise would be added to the input for identification purposes. This allows learning of flexible, over-parametrized models that fit available data well. This makes the proposed identification methods attractive in applications such as guided policy search (GPS)~\cite{levine2013guided, levine2015learning} and non-linear iterative learning control (ILC)~\cite{bristow2006survey}, where they can lead to dramatically decreased sample complexity.

When faced with a system where time-varying dynamics is suspected and no particular
knowledge regarding the dynamics evolution is available, or when the dynamics are known to vary slowly,
a reasonable first choice of algorithm is~\labelcref{eq:smooth}. It is also by
far the fastest of the proposed methods due to the Kalman-filter implementation
of~\cref{sec:kalmanmodel}.\footnote{The Kalman-filter implementation is often
several orders of magnitude faster than solving the optimization problems with
an iterative solver.} Example use cases include when dynamics are changing with a
continuous auxiliary variable, such as temperature, altitude or velocity.
If a smooth parameter drift is found to correlate with an auxiliary variable,
LPV-methodology can be employed to model the dependence explicitly.

Dynamics may change abruptly as a result of, e.g., system failure, change of operating mode,
or when a sudden disturbance enters the system, such as a policy change affecting a market or
a window opening affecting the indoor temperature. The identification method \labelcref{eq:pwconstant} can be employed
to identify when such changes occur, without specifying a priori how many changes are expected.

For simplicity, the regularization weights were kept as simple scalars in this article.
However, all terms $\lambda\normt{\Delta k}^2 = (\Delta k)\T (\lambda I) (\Delta k)$
can be generalized to $(\Delta k)\T \Lambda (\Delta k)$, where $\Lambda$ is an arbitrary positive definite matrix.
This allows incorporation of different scales for different variables with little added implementation complexity.

\section{Conclusions}
We have proposed a framework for identification of linear, time-varying models along trajectories of nonlinear dynamical systems using convex optimization. We showed how a Kalman smoother can be used to estimate the dynamics efficiently in a few special cases, and demonstrated the use of the proposed LTV models on three examples, highlighting their efficiency for trajectory-centric, model-based reinforcement learning, iterative learning control (ILC), and jump-linear system identification. We have also demonstrated the ability of the models to handle non-smooth friction dynamics as well as analyzed the identifiability of the models.


\bibliography{bibtexfile}{}
\bibliographystyle{IEEEtran}
\end{document}